\documentclass{cta-author}
\usepackage{amsmath}
\usepackage{amssymb}
\usepackage{amsfonts}
\usepackage{array}
\usepackage{booktabs}
\usepackage{multirow}
\usepackage{mathptmx}
\usepackage{url}
\usepackage{algorithmic}
\usepackage{algorithm}

\usepackage{graphicx}
\usepackage{subcaption}
\usepackage{balance}

\newtheorem{theorem}{Theorem}{}
{}
\newtheorem{lemma}{Lemma}
{}

\begin{document}

\title{STAIRoute: Early Global Routing using Monotone Staircases for Congestion Reduction}

\author{\au{Bapi Kar$^{1,2}$} \au{Susmita Sur-Kolay$^1$} \au{Chittaranjan Mandal$^2$}}

\address{\add{1}{Indian Statistical Institute, Kolkata, India}
\add{2}{Indian Institute of Technology, Kharagpur, India}
\email{bapi.kar@gmail.com, ssk@isical.ac.in, chitta@iitkgp.ac.in}}

\begin{abstract}
With aggressively shrinking process nodes, physical design methods face severe challenges due to poor convergence and uncertainty in getting an optimal solution. An early detection of potential failures is thus mandated. This has encouraged to devise a feedback mechanism from a lower abstraction level of the design flow to the higher ones, such as placement driven synthesis, routability (timing) driven placement etc. 

Motivated by this, we propose an early global routing framework using pattern routing following the floorplanning stage. We assess feasibility of a floorplan topology of a given design by estimating routability, routed wirelength and vias count while addressing the global congestion scenario across the layout. Different capacity profiles for the routing regions, such as uniform or non-uniform different cases of metal pitch variation across the metals layers ensures adaptability to technology scaling. The proposed algorithm STAIRoute takes $O(n^2kt)$ time for a given design with $n$ blocks and $k$ nets having at most $t$ terminals. Experimental results on a set of floorplanning benchmark circuits show $100\%$ routing completion, with no over-congestion in the routing regions reported. The wirelength for the $t$-terminal ($t\geq$ 2) nets is comparable with the Steiner length computed by FLUTE. An estimation on the number of vias for different capacity profiles is also presented, along with congestion and runtime results.
\end{abstract}


\maketitle

\section{Introduction}
In IC design flow, \textit{global routing} (GR) is indispensable, particularly as an aide to \textit{detailed routing} (DR) of the wires through different metal layers. Shrinking feature dimensions with technological advances in IC fabrication process pose more challenges on the physical design phase. There has been a tremendous increase in routing constraints arising from not only stringent layout design rules but also process variations and sub-wavelength effects of optical lithography. Successful routing completion of the nets without too many iterations or sacrifice in the performance of the designs is thus mandated.
\begin{figure}[!ht]
\centering{
\includegraphics[scale=0.32]{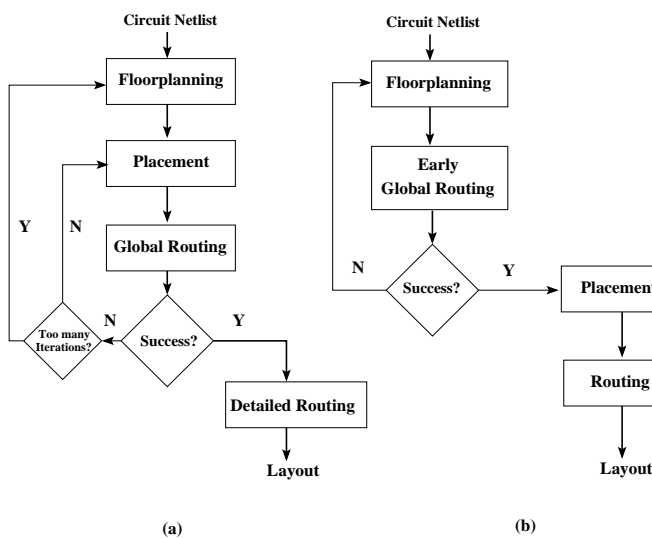}}
\caption{Physical Design Flow: (a) Conventional, and (b) Proposed \label{fig:pdflow}}
\end{figure}

In traditional grid graph based post-placement routing methods, multi-terminal nets are decomposed into two terminal segments using Steiner tree decomposition using \textit{Rectilinear minimum spanning tree} (RMST) \cite{royj}, or \textit{Rectilinear Steiner minimal tree} (RSMT) \cite{panm1} topology as an initial solution with minimum length. Subsequently, congestion driven routing for each two terminal net segment is adopted through the routing regions. The congestion models in those methods have been formulated based on the capacity of the grids and the routing demands through them, along with a penalty function. For any unsuccessful routing due to over congestion ($\geq 100\%$), \textit{Rip-up and Re-reroute} (RRR) techniques using maze routing \cite{lee,sherw} have been applied for possible routing completion while compromising in routed wirelength due to detour. The major challenge in the event of unsuccessful routing is to get back to placement 
stage in order to generate a new placement solution, but with no guarantee for successful routing completion (vide Figure \ref{fig:pdflow} (a)). This may lead to several iterations until the goal is achieved and thus prove to be very costly if the entire design implementation is not completed within a stipulated time frame. In other words, this may have severe impact on time-to-market of the intended design. 

The possibility of recurring iterations at the placement stage (vide Figure \ref{fig:pdflow} (a)) due to failure at global routing stage may however be avoided if we can predict the feasibility of global routing as early as at the floorplanning stage, as depicted in Figure \ref{fig:pdflow} (b). This comprises of the identification of monotone staircase regions as  routing resources while estimating their capacity and formulating the congestion model. These types of routing resources are known to have advantages of \textit{acyclic routing order} 
for successful routing completion \cite{ssk,guru} and avoidance of \textit{switch box routing} \cite{sherw}. They also allow easy \textit{channel resizability} \cite{guru} to mitigate heavy congestion ($\geq 100\%$).

Pattern routing such as single bend (L shaped) \cite{kast}, two bend (Z shaped) \cite{kast,panm1}, or even with more bends such as monotone staircase patterns \cite{zcao,ychang,luj} in global routing has increased significantly in order to find a suitable routing path for a given net. In the recent past, single bend (L shaped) \cite{kast}, two bend (Z shaped) 
\cite{kast,panm1}, or even with more bends such as monotone staircase patterns \cite{zcao,ychang} has gained significant importance in grid-based global routing. With the increasing number of bends, these patterns yield more flexibility in order to find a possible routing path, but at the cost of more vias, if feasible. It was also shown that pattern based routing \cite{kast} is much faster than maze routing, while monotone staircase pattern routing \cite{zcao} has the same time complexity as with Z shaped patterns. A thoughtful trade off between routability (also wirelength) and the number of vias has to be made while keeping in mind that the routing resources are not over congested. Recent work on monotone staircase bipartitioning method \cite{karb3} attempted to address the minimization of the number of vias along a monotone staircase routing path by minimizing the number of bends in it \cite{zcao}. Additionally, the pattern based routing are shown to help in cross talk minimization \cite{kast}.

\subsection{Outline of this work} 
In this paper, we present a new paradign in routability assessment following the floorplanning stage. The proposed routing model uses monotone staircase regions identified in a given floorplan for assessing routing completion without allowing any over-congestion in any of these routing regions. It is important to note that this routing model is different than the grid graph based model used in the post-placement global routers \cite{zcao, kast, mcho1, panm1, royj,ychang,kdai,wliu}. The outline of the proposed early global routing method, we call it STAIRoute, presented in Figure \ref{fig:view} as follows: 
\begin{enumerate}
 \item \textit{Identification of routing regions} in a given floorplan by using a monotone staircase bipartitioning algorithm such as \cite{karb,karb3}; 
 \item \textit{Graph theoretic} formulation based on these regions is used to determine a feasible routing paths for the given nets; uses congestion aware layer-assignment approach.
 \item \textit{Routing order} of the nets based on half perimeter wire length (HPWL) and the number of terminals (Netdegree);
 \item \textit{Multi-terminal net decomposition} to identify a set of two-terminal net segments using minimum spanning tree algorithm; obtained a new \textit{Steiner tree} topology;
 \item \textit{Routing through a number of metal layers} using a shortest path algorithm to find an acceptable routing path;
 \item \textit{Ensuring congestion} in the routing regions is restricted to $100\%$ across the metal layers;
 \item \textit{Supporting} both unreserved and reserved layer based routing; provides an estimation on the number of vias for possible layer change along the routing paths.
\end{enumerate}

\begin{figure}[!ht]
\centering{
\includegraphics[scale=0.31]{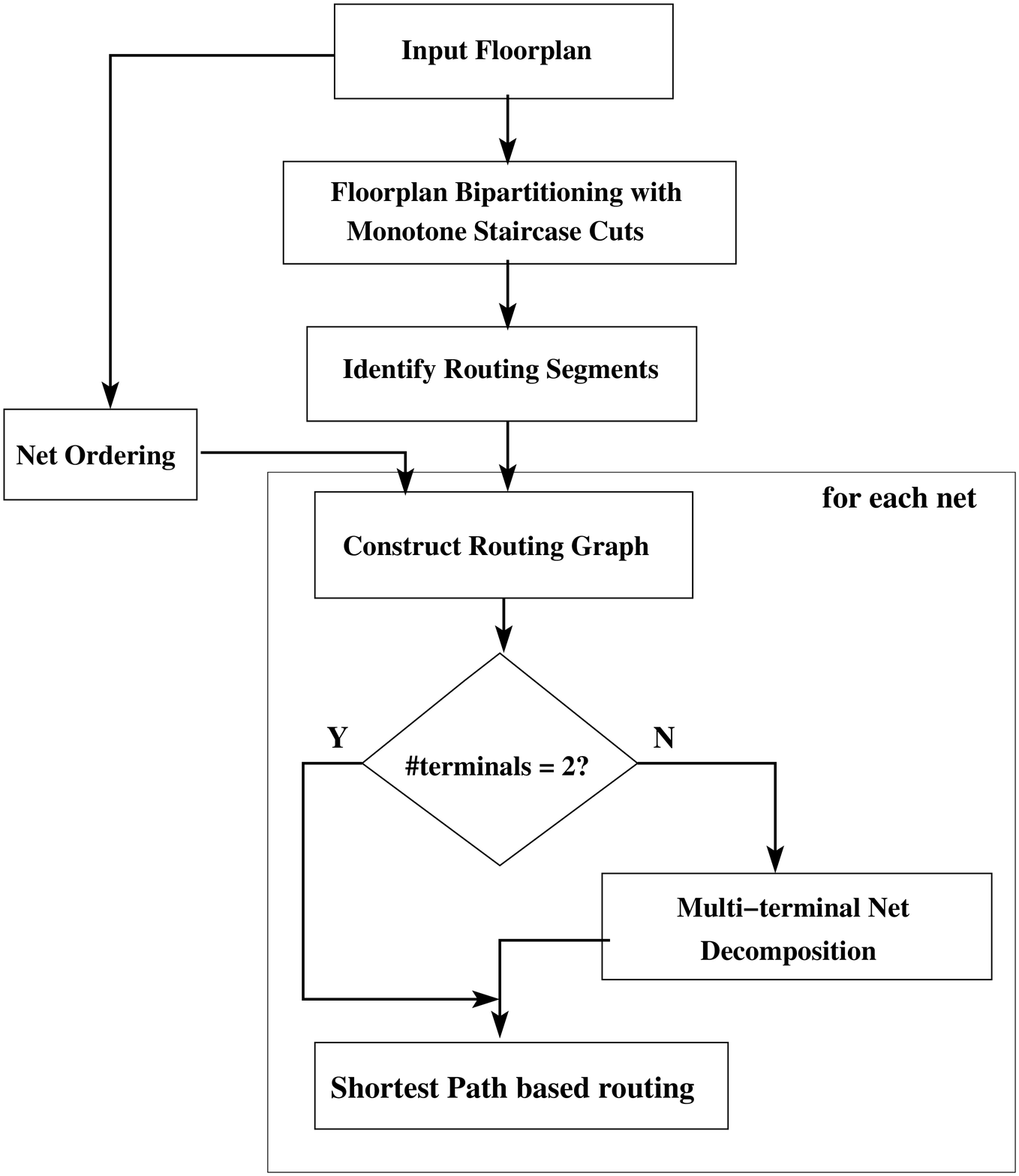}}
\caption{Outline of the proposed early global routing framework \label{fig:view}}
\end{figure}

This paper is organized as follows: Section \ref{sec:prelim} revisits the preliminaries of monotone staircase bipartitioning method followed by Section \ref{sec:work} that includes related topics and the proposed global routing method using monotone staircases as the routing resources. Results are presented in Sections \ref{sec:result}, and the concluding remarks in \ref{sec:discuss}.

\section{Preliminaries on Monotone Staircase Regions}
\label{sec:prelim}
In this section, we review the monotone staircase regions in a given floorplan and the floorplan birpartitioning framework to obtain them. Methods for top-down hierarchical monotone staircase bipartitioning of floorplans, both in \textit{Area-balanced} and \textit{Number-balanced} bipartition appear in \cite{dasg,karb,karb3,majum1,majum2}. \textit{Area-balanced} bipartition is employed when the area of the blocks in a given floorplan have significant variance whereas \textit{Number-balanced} bipartition is applicable for negligible variance in the area of the blocks. In \cite{majum1,majum2}, the balanced bipartitioner used iterative max-flow based \cite{yang} min-cut algorithm and thereby incurred higher time complexity at each level of the hierarchy. In \cite{dasg}, emphasis has been given merely to the hierarchical number balanced monotone staircase bipartitioning using depth-first traversal method in linear time at a given level of the hierarchy. 
\begin{figure}[!ht]
\centering{
\includegraphics[scale=0.34]{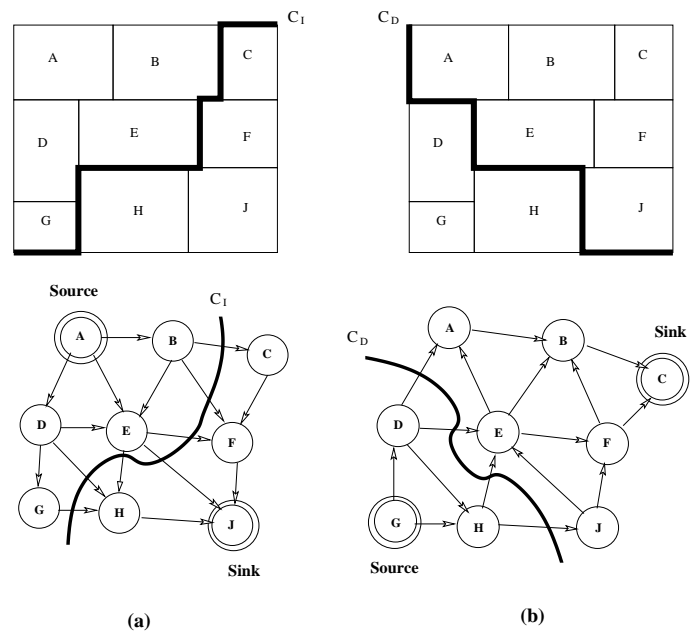}}
\caption{A floorplan with staircase regions and the corresponding ms-cuts in its block adjacency graph (BAG): (a) monotone increasing staircase (MIS), (b) monotone decreasing staircase (MDS) \cite{karb} \label{fig:mis}}
\end{figure}

Recently, a faster yet more accurate top-down hierarchical monotone staircase bipartition \cite{karb} has been proposed to generate monotone staircase cuts, abbreviated as \textit{ms-cut} as we subsequently refer to it. This algorithm takes $O(nk\log n)$ time and also ensures ms-cuts of increasing (decreasing) orientation at alternate levels of the hierarchy (vide Figure \ref{fig:graph} 
(a) and (b)), namely \textit{MSC tree}. 

In order to identify monotone increasing (decreasing) staircase regions ($C_I$ ($C_D$) as depicted in Figure \ref{fig:mis} (a) ((b)), abbreviated as \textit{MIS (MDS)}, for a given a planar embedding of a floorplan topology with $n$ blocks, an unweighted directed graph \cite{karb,majum1}, called \textit{block adjacency graph} (BAG) $G(V_b,E_b)$ is formulated. The graph is defined as follows: $V_b$ = $\{b_i \vert  \forall$ blocks $b_i$ in the floorplan$\}$ and $E_b = \{ <b_i,b_j> \vert$ block $b_i$ is either on the left of or above (below) its adjacent block $b_j\}$. Note that $|V_b|$ = $n$ and 
$|E_b|$ = $3(n-1)$ (vide Lemma \ref{lemma2}).

\begin{lemma}
\label{lemma1}
Given a floorplan with $n$ blocks, its MSC tree $(V_m,E_m)$ corresponding to the set $C$ of monotone staircase regions has $n-1$ ms-cuts (internal nodes).
\end{lemma}

\begin{proof}
In a full binary tree, an internal node has two children (out degree = $2$) whereas an external (leaf) node has no children (out degree = 0). In our case, the internal nodes correspond to the ms-cuts in the MSC tree, and the external nodes are the blocks in the given floorplan.\\
Hence $\sum_{i}${OutDeg($v_i$)} = $|E_m|$ = $|V_m| -1$, where $T_m$ = $G(V_m,E_m)$ is the resulting MSC tree as shown in Figure \ref{fig:graph}(b).\\
$\Rightarrow$ $2*|C|$ + $0*n$ = ($|C|$ + $n$) -1; where $|C|$ and $n$ are the number of ms-cuts and blocks respectively.\\
$\Rightarrow$  $|C|$ = $n-1$.\\
\end{proof}

\section{This Work} 
\label{sec:work}
The contribution in this paper is to propose a new routing model based on the floorplan bipartitioning results for a given design. Then we showcase how this routing model is used to estimate several routability metrices such as routing completion, routed wirelength, via count and the global congestion scenario, for a given number of metal layers. Following sections discuss about this work in detail.

\subsection{Routing Region Definition}
Using the recent hierarchical monotone staircase bipartition framework \cite{karb}, we obtain a set of MIS (MDS) regions $C$ = \{$C_i$\} at alternate levels of the hierarchy in MSC tree. For the rest of the paper, we refer \textit{region} and \textit{segment} to monotone staircase region and its rectilinear segment respectively. These regions are used as the routing resources for the proposed early global routing framework. 

Each region consists of one or more rectilinear segments, bounded by a distinct pair of blocks. For each segment in a region, the number of nets to be routed through it is estimated from the number of cut nets in the respective ms-cut node in the MSC tree. This is denoted as \textit{reference capacity} $r_k$ for the $k$th segment. At any point during routing, its capacity \textit{usage} $u_k$ (also known as routing demand in the existing literature), gives the region utilization.  
\begin{figure*}[!ht]
\centering{
\includegraphics[scale=0.61]{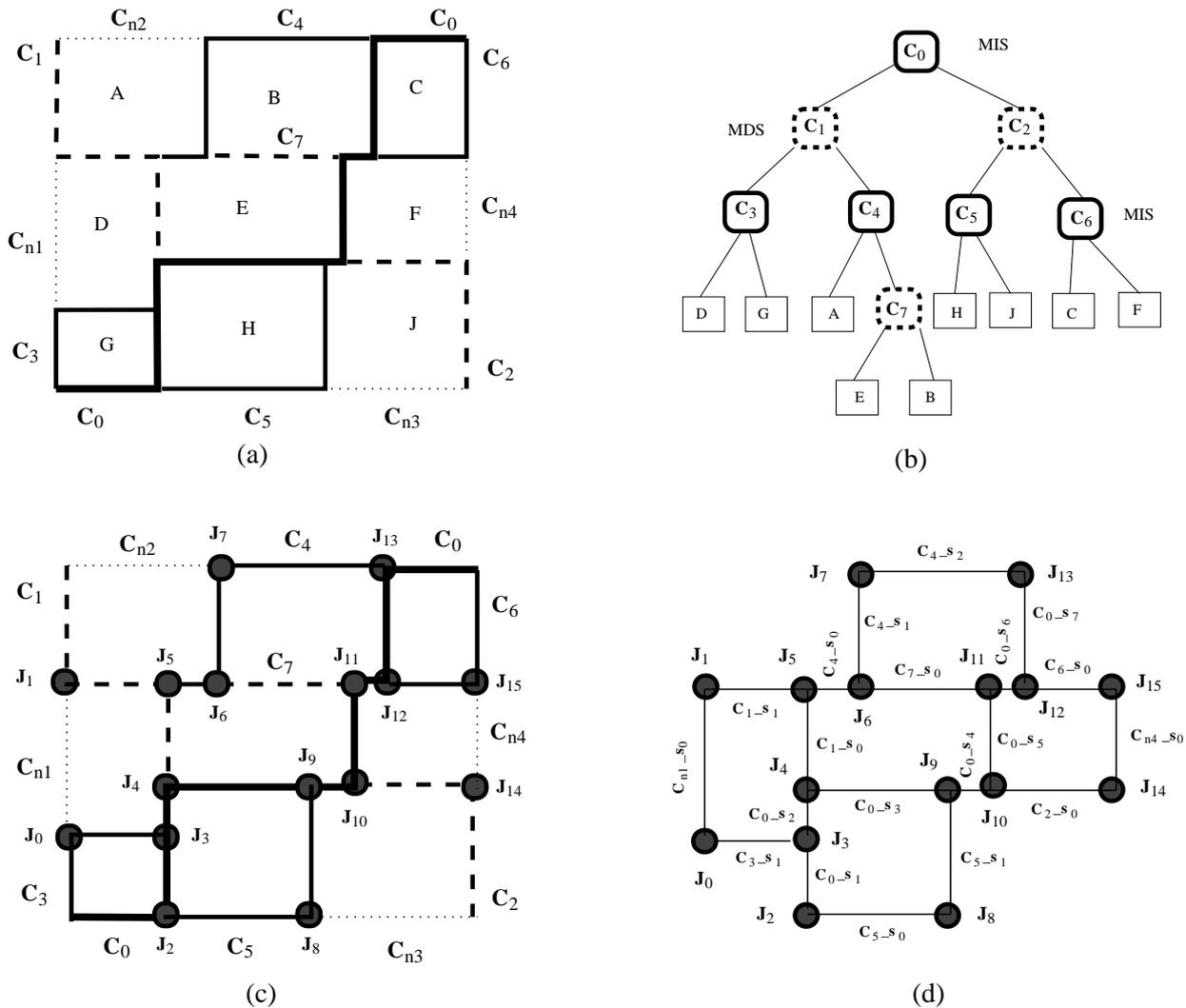}}
\caption{A floorplan along with: (a) its hierarchy of monotone staircase regions, (b) the corresponding MSC tree, (c) T-junctions at which staircases intersect, and (d) the corresponding junction graph \label{fig:graph}}
\end{figure*}

As in Figure \ref{fig:mis} (a), the highlighted ms-cut $C_I$ on BAG contains seven cut edges \{$GH$, $DH$, $EH$, $EJ$, $EF$, $BF$ and $BC$\} corresponds to an MIS region $C_I$ having seven segments. Additionally, it has two more horizontal segments: one for the bottom side of the block $G$ at the bottom-left corner with the boundary of the floorplan while the other is on the top side of the block $C$ at the top-right corner of the floorplan. Likewise, an MDS region $C_D$ having six segments is illustrated in Figure \ref{fig:mis} (b). Overall, a flooplan along with a set of MIS/MDS regions $C_0$ to $C_7$ is shown in Figure \ref{fig:graph} (a). Regions with one segment having either vertical or horizontal orientation are termed as \textit{degenerated} staircase regions, $C_7$ being such an example of such a degenerated region. 

However, there exist a few more isolated segments along the boundary of the floorplan that are not identified as part of the MSC tree generation, and can be termed as \textit{non-MS regions}. In Figure \ref{fig:graph}(a), $C_{n1}$ to $C_{n4}$ are the examples of few such regions. Their routing capacity $r_k$ is computed based on the number terminals on it, and those with nonzero $r_k$ contribute to routing as valid routing resources.

\subsection{Routing Model: the Junction Graph (JG)}
\label{sub:jg}
Now, we present our routing model using a set of monotone staircase regions and their intersection points, called T-junctions (vide Figure \ref{fig:graph} (c)). It is evident that there exists a segment between each pair of adjacent T-junctions, henceforth referred as $junctions$.

\begin{lemma}
\label{lemma2}
Given a floorplan with $n$ blocks, the number of T junctions in it is $2n-2$.
\end{lemma}
\begin{proof}
Every internal face in BAG corresponds to a T-junction, and is bounded by $3$ edges. Thus we have $3(f-1)$ = $2m$ excluding the exterior face, where $f$ and $m$ being the number of faces and edges in BAG respectively. Using \textit{Euler formula} of $n-m+f=2$ for planar graphs, and replacing $f$ by $2m/3+1$, we get $m = 3(n-1)$.\\
Hence, the number of T-junctions in the floorplan = $f-1$ = $2m/3$  = $2n-2$.
\end{proof}

Using the notion of T-junctions, we construct a weighted undirected graph (vide Figure \ref{fig:graph}(d)), called \textit{junction graph}, $G_j$ = ($V_j$,$E_j$), where $V_j$ = $\{J_p\}$ corresponds to a set of junctions, and $E_j$ = \{\{$J_p$,$J_q$\} $|$ as a pair of adjacent junctions $\{J_p, J_q\}$ with a segment $s_k$ of a region $C_m \in C$ between them\}. As depicted in Figure \ref{fig:graph} (c), all the junctions, except those near the corners of the floorplan with degree two, have degree of three in $G_j$, i.e., have edges with three adjacent junctions. Using Lemma \ref{lemma2}, it can be shown that $|E_j|$ = $3n-7$. 

The weight of each edge $e_{pq} \in E_j$ is computed as:
\begin{equation}
\label{eqn1}
wt(e_{pq}) = length_(s_k)/(1-p_k)
\end{equation}
where $p_k$, the \textit{normalized usage} through the segment $s_k$, is defined as:
\begin{equation}
\label{eqn1a}
p_k = u_k/r_k
\end{equation}
And, we define $(1-p_k)$ as the \textit{usage penalty} on the edge weight for routing a net through the corresponding segment $s_k$. In Figure \ref{fig:cong}, we illustrate the variation of edge weight with respect to the \textit{normalized usage} $p_k$. 
\begin{figure}[!ht]
\centering{
\includegraphics[scale=0.42]{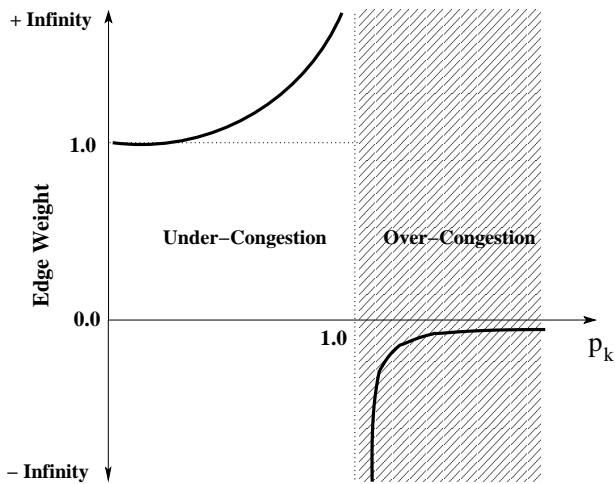}}
\caption{Edge weight ($wt(e_{pq})$) vs. normalized usage ($p_k$) \label{fig:cong}}
\end{figure}

In the proposed routing framework, congestion is avoided in all the segments by restricting $p_k$ to be no more than $1.0$. This is achieved by setting the edge weight to \textit{Infinity} whenever $p_k = 1.0$. The corresponding edge is virtually removed from $E_j$. This ensures that the case of $p_k > 1.0$ does not occur. In Figure \ref{fig:cong}, we mark the regions $p_k \leq 1.0$ and $p_k > 1.0$ as \textit{Under-Congestion} and \textit{Over-Congestion} regions respectively. Therefore, we restrict to \textit{Under-Congestion} while formulating the global routing graph such that there is no congestion in any of the routing resources. However, it may be noted that routing may fail for some of the nets due to insufficient capacity of some of the routing resources for a specified number of metal layers. 

\begin{lemma}
\label{lemma3a}
The construction of the junction graph takes $O(n)$ time.
\end{lemma}

\begin{proof}
By Lemma \ref{lemma2}, we know that there are $O(n)$ edges in the BAG, where each edge corresponds to a segment. Therefore, for each segment $s_k$ having a pair of junctions \{$J_p$,$J_q$\} as its endpoints, an edge is inserted in the $G_j$. Hence, the construction of the junction graph $G_j$ takes $O(n)$ time.
\end{proof}

\subsubsection{Routing Model: the Global Staircase Routing Graph (GSRG)}
In this section, we present our proposed global routing framework by extending the junction graph $G_j$ for each net.
\begin{figure*}[!ht]
\centering{
\includegraphics[scale=0.48]{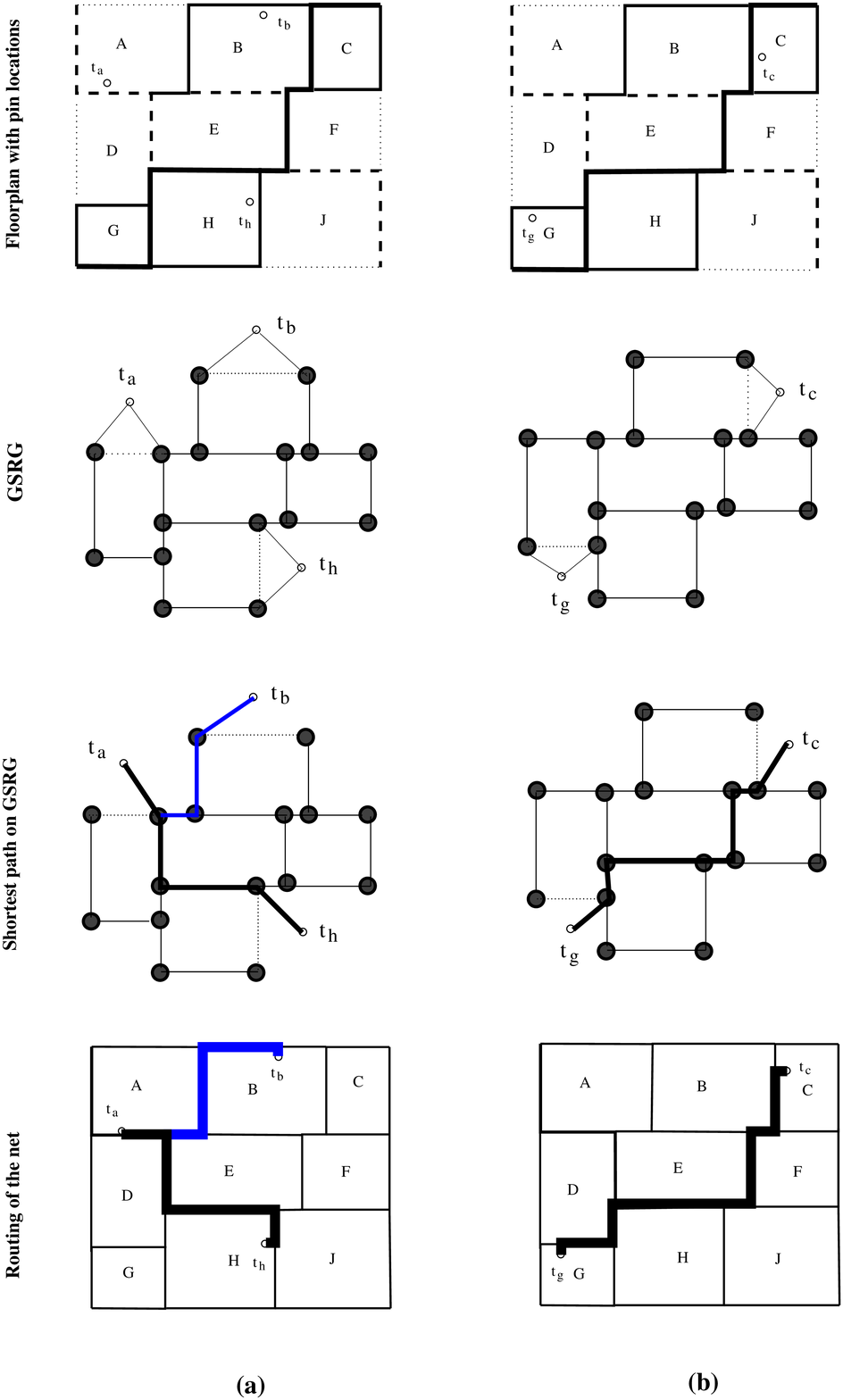}}
\caption{Steps for constructing the Global staircase routing graph (GSRG) from the Junction graph for (a) a $3$-terminal net $n_i = \{t_a,t_b,t_h\}$, and (b) a $2$-terminal net $n_j = \{t_c,t_g\}$, along with corresponding routed paths in GSRG and finally routed nets in the floorplan topology \label{fig:gsrg}}
\end{figure*}

Let $N$ be a set of nets for a given floorplan. For each $t$-terminal ($t\geq2$) net $n_i \in N$, we use $G_j$ as the backbone to derive the corresponding \textit{Global Staircase Routing Graph} (GSRG) as depicted in Figure \ref{fig:gsrg} (a) and \ref{fig:gsrg} (b). The GSRG is defined as $G_{r}^i$ = ($V_{r}^i$,$E_{r}^i$), where $V_{r}^i$ = $V_j$ $\bigcup$ \{$t_l\vert t_l\in n_i$\}, and $E_{r}^i$ = $E_j \bigcup E_{lp}$. Each pin-junction edge $e_{lp} \in E_{lp}$ is defined as $e_{lp}$ = \{$t_l$,$J_p$\} $|$ $\forall t_l \in n_i$ and $\exists J_p \in J$, the pin $t_l$ resides on a segment $s_k$ associated with the junction $J_p$\}. As before, we calculate the weight of a pin-junction edge $e_{lp}$ as:
\begin{equation}
\label{eqn1b}
wt(e_{lp}) = \mathrm{distance}(t_l,J_p)/(1-p_k).
\end{equation}
and define $(1-p_k)$ as the \textit{usage penalty} on the edge weight for 
routing a net through the corresponding segment $s_k$.

\begin{lemma}
\label{lemma3b}
For a $t$-terminal net, the construction of its GSRG takes $O(t)$ time.
\end{lemma}

\begin{proof}
As defined, the GSRG $G_{r}^i$ = ($V_{r}^i$,$E_{r}^i$) for a given net $n_i$ with $t$ terminals is obtained by augmenting the junction graph $G_j$ = ($V_j,E_j$). In other words, $V_j$ is extended by $t$ terminals connected to $n_i$ in order to obtain $V_{r}^i$. It is also to be noted that each terminal resides on a segment $s_k$, having a pair of junctions ($J_p,J_q$) on either ends. Therefore, each terminal (pin) contributes $2$ pin-junction edges and thus total $2t$ edges to $G_{r}^i$ for all $t$ terminals. Hence, the construction of $G_{r}^i$ takes $O(t)$ time for each net.
\end{proof}

After routing a net $n_i$ successfully, we update $u_k$ for all such segments $s_k$ through which $n_i$ is routed. Subsequently, the weights of the edges in $G_j$ are updated before we route the subsequent net $n_{i+1}$. When congestion is about to occur in a given segment ($p_k = 1$), the weight of the corresponding edge in $G_{r}^i$ becomes \textit{Infinity}. No routing is possible through such segments and the relevant edges virtually disappear making $G_{r}^i$ more sparse after each iteration of routing. To summarize, the \textit{normalized usage} $p_k$ in this framework is  constrained to a maximum of $100\%$, thus restricting the number of routed nets ($u_k$) through a given segment to be no more than its capacity ($r_k$).

In order to extend this model for $M (\geq 1)$ metal layers, we keep a parameter called \textit{currLayer($s_k$)} associated with each segment $s_k$, initialized to $1$ and can go up to a maximum of $M$ metal layers. When congestion is about to occur in $s_k$ ($p_k = 1$), we increment \textit{currLayer($s_k$)} to the subsequent metal layer. Here the subsequent metal layer has different implication in \textit{(un)reserved layer model}; the subsequent layer can either be one layer above \textit{currLayer($s_k$)} or the next permitted layer based on the particular (horizontal/vertical) orientation of $s_k$ in the corresponding reserved layer model. This means that the resource $s_k$ has exhausted its entire capacity (i.e. $u_K = r_k$) for the current metal layer and is now ready for routing the nets through it for the next metal layer restricted by $M$. In this regard, the variation of $r_k$ across the metal layers (up to $M$) plays a significant role and thus directly impacts the routing completion of all the nets.
\begin{figure}[!ht]
\begin{subfigure}[b]{0.5\textwidth}
\centering{
\includegraphics[scale=0.47]{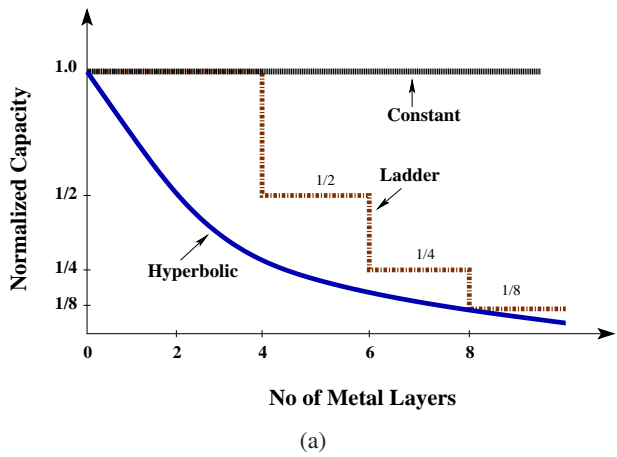}}
\caption{}
\end{subfigure}
\begin{subfigure}[b]{0.5\textwidth}
\centering{
\includegraphics[scale=0.6]{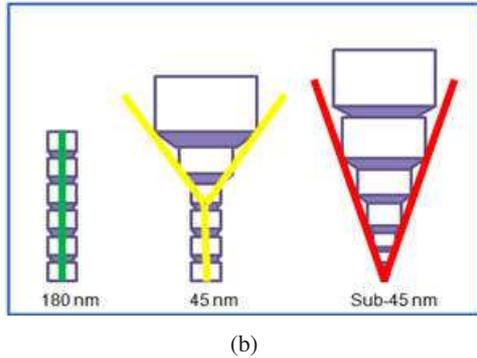}}
\caption{}
\end{subfigure}
\caption{(a) Normalized Capacity profile of a routing resource ($s_k$) vs number 
of metal layers ($M$), and (b) Metal layer variation across process nodes 
\cite{m_pitch}}
\label{fig:capscale}        
\end{figure}

In Figure \ref{fig:capscale} (a), we study different scenario of uniform as well as varying capacity profile for all the routing resources $s_k$ across the metal layers. In case of \textit{Uniform} profile, $r_k$ carries the same value across the metal layers. We consider two different varying capacity profiles, one is \textit{Hyperbolic} ($1/M$) type pattern, while the other being a 
\textit{Ladder} type pattern. In case of the former, $r_k$ is more aggressively scaled across the metal layers, the latter is a more realistic scenario that captures the latest trend of the metal pitch/width variation across the metal layers in the recent nanometer technologies (vide Figure \ref{fig:capscale} (b) \cite{m_pitch}). 

\subsection{Multi-terminal Nets and Their Two-terminal Decomposition}
\label{sub:mnet}
In a global routing framework, routing a $t (>2)$-terminal net is crucial and obtaining an efficient solution for minimal length is a hard problem. Several works have been done so far to obtain the best possible $t-1$ net segments for a $t$-terminal net such as \textit{Rectilinear Steiner Minimal Tree} (RSMT) topology proposed in \textit{FLUTE} \cite{cchu} based on a well defined grid structure known as \textit{Hanan grid} \cite{hana,sherw}. Since the proposed work is based on a gridless framework and the routing regions are aligned with the MIS/MDS regions, we cannot adopt any grid-based RSMT framework such as \textit{FLUTE} \cite{cchu}.

Therefore, we propose a new method for multi-terminal net decomposition suitable for the proposed global routing framework. We construct a complete undirected graph for a given $t$-terminal ($t>2$) net $n_i \in N$, $G_{c}^i$ = ($V_{c}^i$,$E_{c}^i$) such that $V_{c}^i$ = \{$t_k$\}, $\forall t_k \in n_i$ and $E_{c}^i$ = \{\{$t_j$,$t_k$\} $|$ $\forall t_j,t_k \in n_i$ and $t_j \neq t_k$\}. The weight of each edge $e_{jk}$  = \{$t_j$,$t_k$\} $\in E_{c}^i$ is computed as \textit{half the perimeter length} (HPWL) of the bounding box for each terminal pair ($t_i,t_j$) (vide  Figure \ref{fig:mnet} (a)). It is evident that $|V_{c}^i|$ = $O(t)$ and $|E_{c}^i|$ = $O(t^2)$. By employing $O(n^2)$ Prim's Minimum Spanning Tree (MST) algorithm \cite{cormen}, we obtain a \textit{minimum spanning tree} (MST) $T^i$ for $G_{c}^i$ having $t-1$ edges, i.e., $t-1$ valid $2$-terminal pairs. For each edge $e_{jk}$ = \{$t_j$, $t_k$\} $\in T^i$, we perform 2-terminal net routing by applying Dijkstra's single source shortest path algorithm \cite{cormen}. Once we obtain the routing for all such terminal pairs, we obtain the \textit{Steiner points} by identifying the common routing segments as illustrated by an example in Figure \ref{fig:mnet}. 
\begin{figure}
\centering{
\includegraphics[scale=0.3]{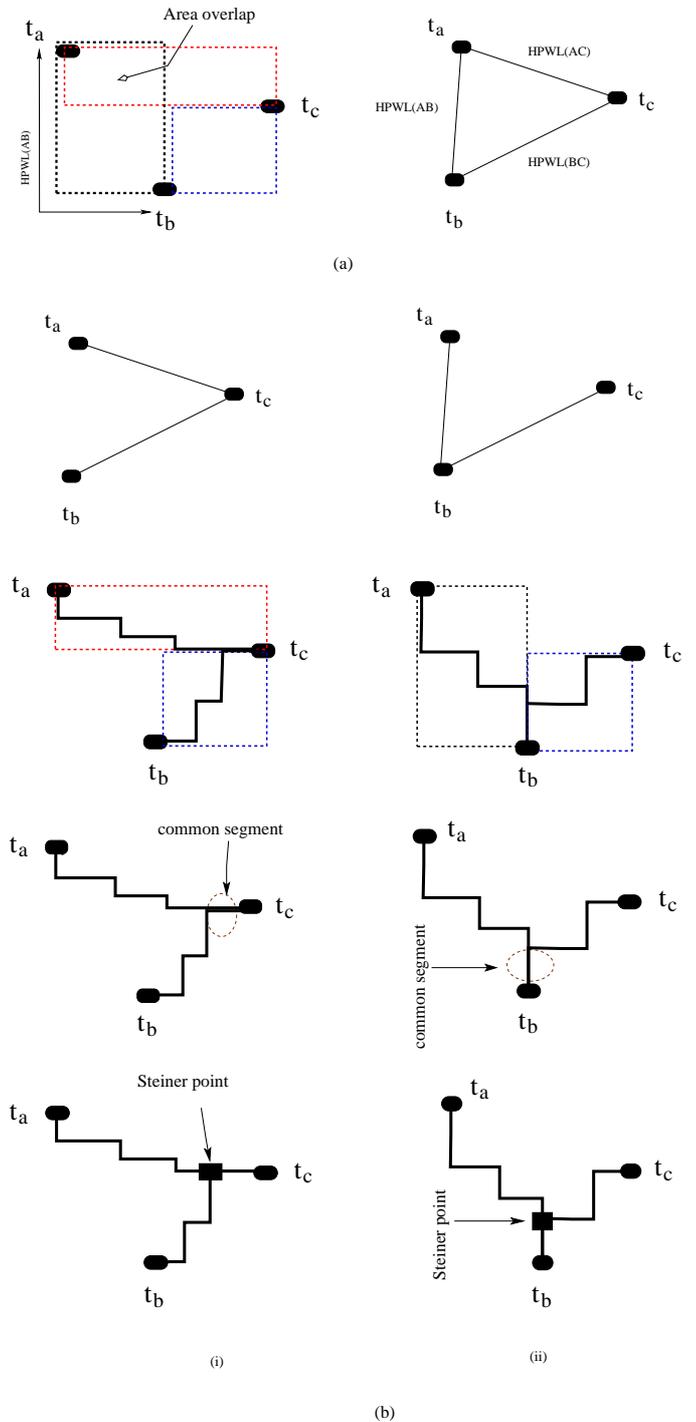}}
\caption{Multi-terminal net decomposition and Staircase Minimal Steiner Tree \label{fig:mnet}}
\end{figure}

We consider an example of a $3$-terminal net $n_1$ with terminals $\{t_a,t_b,t_c\}$ to illustrate the proposed $2$-terminal net decomposition as shown in Figure \ref{fig:mnet}. In this case, $G_{c1}$, a $3$-clique, has $3$ vertices $\{t_a,t_b,t_c\}$, and $3$ edges, namely $\{t_a,t_b\}$, $\{t_b,t_c\}$ and $\{t_a,t_c\}$, along with their corresponding edge weights (vide Figure \ref{fig:mnet} (a)). As shown in Figures \ref{fig:mnet} (b)-(i) and (b)-(ii), only one of the instances of minimum spanning tree $T_{c1}$ is greedily obtained by the said MST algorithm as the final solution. 

Depending on a specific $T_{c1}$ thus obtained, the proposed $2$-terminal net segment routing, presented in the next section, for each valid terminal pair is applied. Once the routing for all the designated terminal pairs are obtained, we identify the \textit{Steiner points} similar to the state-of-the-art grid-based multi-terminal net decomposition methods (FLUTE \cite{cchu}), as illustrated in 
Figure \ref{fig:mnet} (b). The main difference is that this work is based on a gridless framework using monotone staircase regions as the routing resources. This topology may be termed as \textit{Staircase Minimal Steiner Tree} (SMST).

\subsection{STAIRoute: the proposed algorithm}
We present the proposed global routing algorithm \textit{STAIRoute} using monotone staircase regions in Algorithm \ref{alg:1}. This algorithm takes two inputs, namely a ordered set of nets $N$ and the junction graph $G_j$ as defined in Section \ref{sub:jg}. For each net $n_i \in N$, the GSRG $G_{r}^i$ is constructed and a routing path for the net $n_i$ is obtained by applying a shortest path algorithm on $G_{r}^i$. We have implemented $O(n^2)$ Dijkstra's shortest path algorithm \cite{cormen}, namely \textit{DijkstraSSP()}, presented in Algorithm \ref{alg:1}. 

\begin{figure*}[!ht]
\centering{
\includegraphics[scale=0.42]{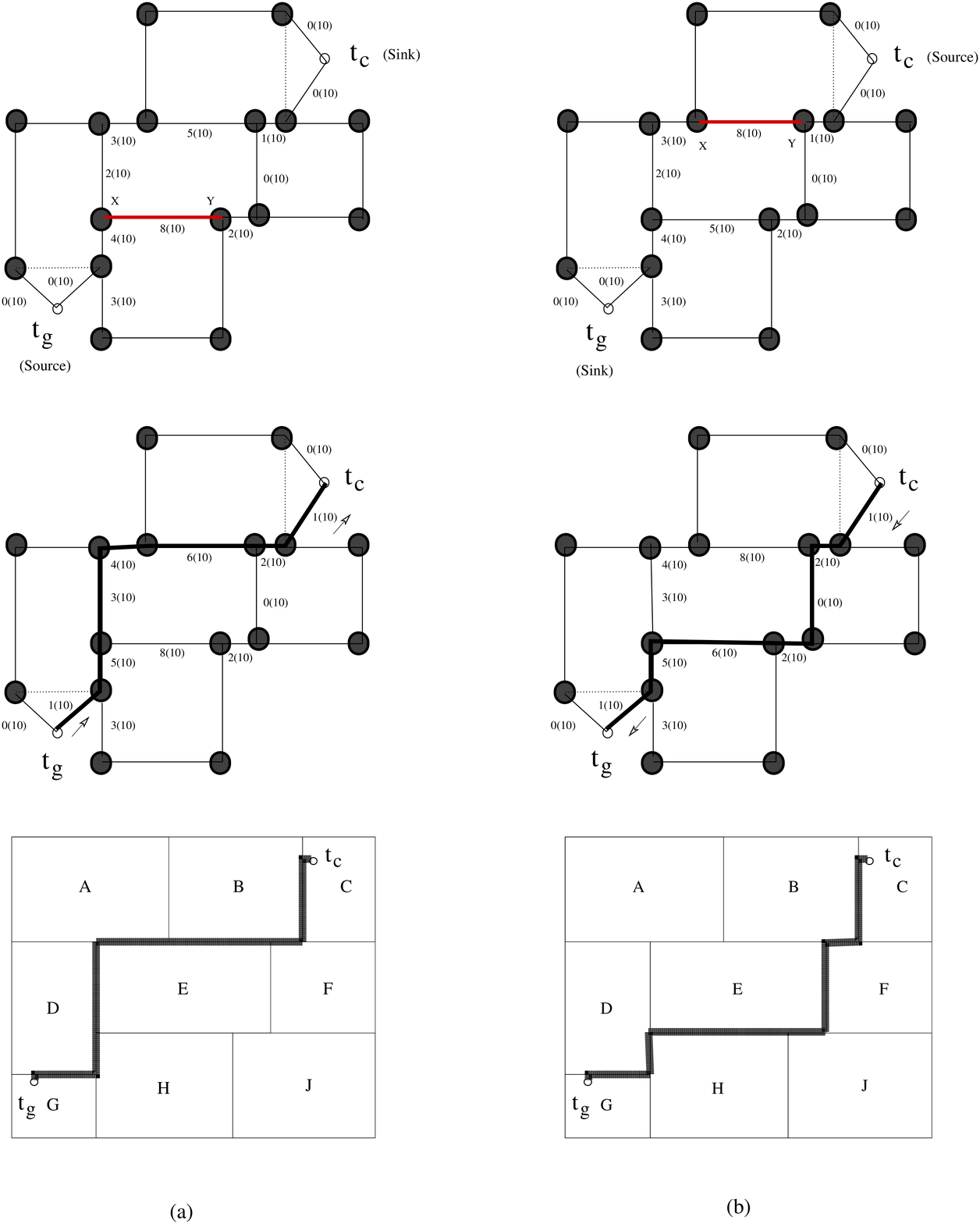}}
\caption{Exploring a routing path based on: (a) Forward Search, and (b) Backward Search \label{fig:routedir}}
\end{figure*}

For each $2$-terminal net (segment), we consider two cases of identifying the source vertex between a pair of terminals before we apply the shortest path algorithm as: 
\begin{enumerate}
 \item the minimum $x$ coordinate (or the minimum $y$ coordinate in case both the terminals have the same $x$ coordinate)
 \item the maximum $x$ coordinate (or the maximum $y$ coordinate in case both have the same $x$ coordinate)
\end{enumerate}
and the procedure \textit{IdentifySource()} in Algorithm \ref{alg:1} is used for that purpose.

We term them as \textit{Forward} (FWD) and \textit{Backward} (BACK) search respectively. In Figure \ref{fig:routedir} (a) and (b), we illustrate the respective cases for a $2$-terminal net \{$t_g,t_c$\} and show that both search procedures can potentially give different routing paths. One may have a potentially better solution than the other in terms of routability, congestion scenario along with wirelength, and finally via count. The variation in wirelength due to FWD (BACK) search arises when certain resource(s) along the respective paths are fully utilized in a given metal layer; with the possibility of switching to the next available metal layer if permitted, leads to increase in the via count. Otherwise, the routing path is detoured beyond the bound box of the terminals, leading to increase in length. As long as the alternatives paths remain confined within the bounding box of the terminals, there is no variation among the respective wirelengths. 
\begin{algorithm}[!ht]
\caption{STAIRoute}
\label{alg:1}
\begin{algorithmic}
\REQUIRE{$G_j(V_j$,$E_j$), Ordered nets $N$}
\ENSURE{Global routing for each $t$-terminal ($t\geq2$) nets ($n_i \in N$) with 
100\% routability and usage $\leq$ 100\%}
\newline
\FORALL{sorted nets $n_i \in N$}
\STATE $G_{r}^i$ = ConstructGSRG($G_j$,$n_i$)\\
\IF {Netdegree($n_i$) == 2}
\STATE /*Netdegree($n_i$) = Number of terminals in $n_i$*/\\
\STATE Source = IdentifySource($n_i$.terminals) /*for Forward or Backward search 
(vide Fig. \ref{fig:routedir})*/\\
\STATE Path(Source,Sink) = DijkstraSSP($G_{r}^i$,Source)
  \IF{There exists a routing path from Source to Sink}
  \STATE $n_i$ is routed.
  \STATE Update $u_k$ for the respective segments.
  \STATE WireLength($n_i$)
  \STATE ViaCount($n_i$)
  \ELSE
  \STATE Routing $n_i$ is a failure and continue for $n_{i+1}$
  \ENDIF
\ELSE
\STATE $G_{c}^i$ = ConstructNodeClique($n_i$.terminals)\\
\STATE $T^i$ = ObtainMST($G_{c}^i$) /*described in Section \ref{sub:mnet}*/
  \FORALL{edges $(t_j,t_k) \in T^i$}
  \STATE Source = IdentifySource($t_j,t_k$) /*for Forward or Backward search 
(vide Fig. \ref{fig:routedir})*/\\
  \STATE Path(Source,Sink) = DijkstraSSP($G_{r}^i$,Source)
    \IF {There exists a routing path from Source to Sink}
    \STATE $2$-terminal net segment is routed; calculate the segment length.
    \STATE update the $u_k$ for the respective segments.
    \ELSE
    \STATE Routing $n_i$ is a failure and continue for $n_{i+1}$
    \ENDIF
  \ENDFOR
\STATE Identify the \textit{Steiner Point(s)} /*vide Fig. \ref{fig:mnet}*/\\
\STATE WireLength($n_i$)
\STATE ViaCount($n_i$)
\ENDIF 
\ENDFOR
 \end{algorithmic}
\end{algorithm}

In the unreserved layer model, routing a net may incur a number of vias due to change in metal layer used to route through a set of routing resources. It does not depend on their vertical/horizontal orientation. For reserved layer model, each direction (vertical or horizontal) supports only a set of designated routing layers, e.g., horizontal layers are routed through odd numbered layers only. The number of vias along a routing path depends on the number of bends in it, i.e., the alternating (horizontal/vertical) orientation of the contiguous routing resources, for a minimum change of one metal layer among the resources along that path \cite{karb3,royj}. Congestion in these regions may also contribute to the number of vias along a routing path, in both the cases. From the example shown in Figure \ref{fig:viacount} (a) and (b), we notice that the routing path for a given net ($t_g$, $t_c$) needs $3$ and $5$ vias for FWD and BACK searches respectively. Therefore, depending on the netlist and the floorplan topology of a given circuit, one method may dominate over the other. This method can be extended to $t$ ($>2$)-terminal nets, since we decompose those nets using the method stated earlier into $2$-terminal net segments. A better routing path for each of the resulting net segments can be obtained while employing either of the search procedures at a time.
\begin{figure*}[!ht]
\centering{
\includegraphics[scale=0.65]{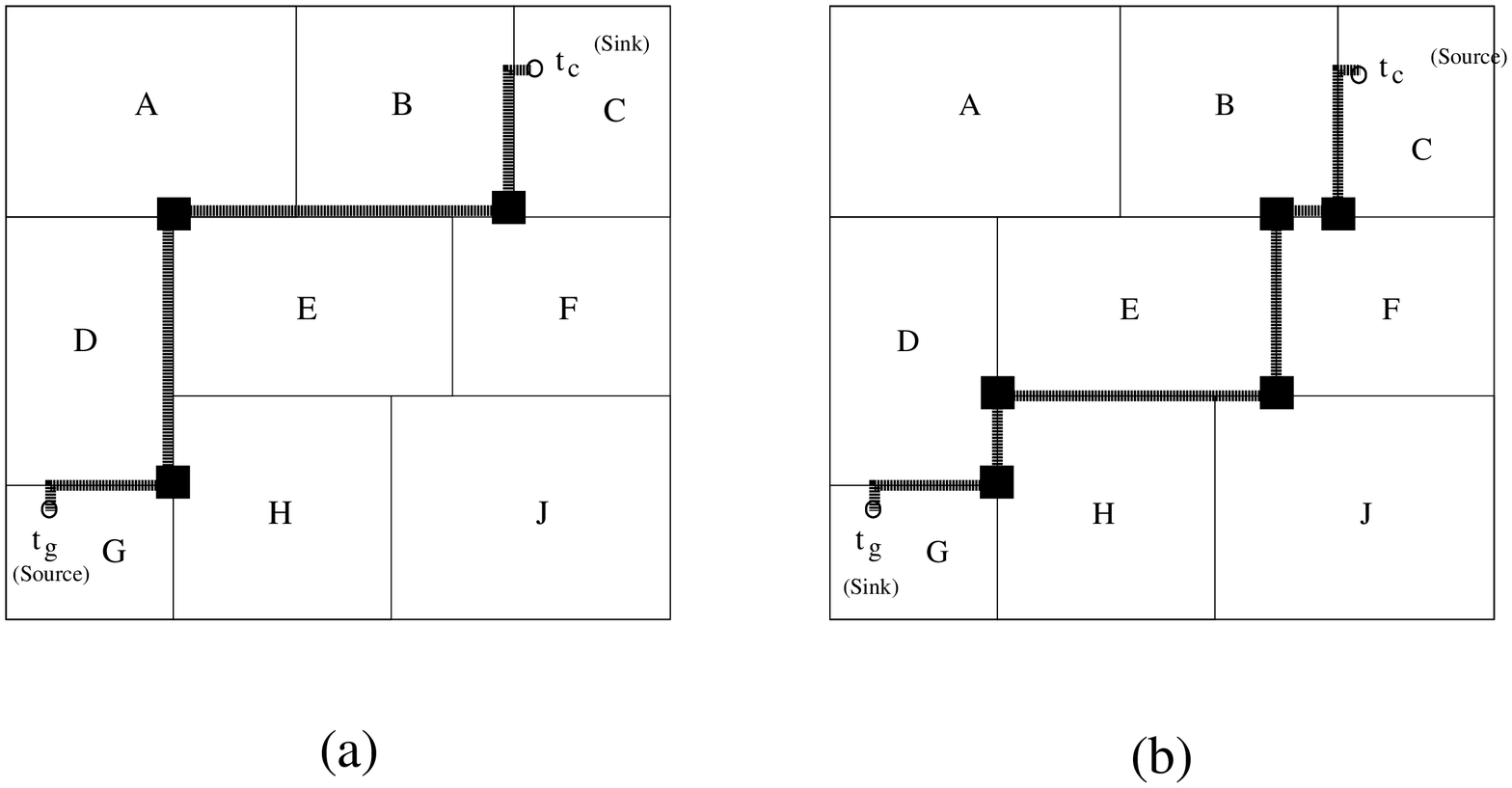}}
\caption{Impact of search direction on via count along a routing path: (a) Forward Search with $3$ vias, and (b) Backward Search with $5$ vias \label{fig:viacount}}
\end{figure*}

Before the routing procedure starts, it is very important that the all nets ($n_i \in N$) are \textit{ordered} based on their \textit{half perimeter wire length} (HPWL), and the number of terminals in each net, i.e., (\textit{Netdegree}). The net ordering (priority) is determined based on the non-decreasing order of HPWL first and then Netdegree. A net with smaller HPWL and then Netdegree, has the precedence over other nets. The aim is to ensure that the shorter (local) nets are routed before the longer ones so as to avoid congestion in the routing resources as well as have a uniform routing distribution across the layout of the design without incurring too much detours. Here detoures imply that the routing is done along a non-monotone path such as U-shaped paths. We illustrate the working of this algorithm for $t$ ($\geq 2$)-terminal nets in Figure \ref{fig:gsrg}.

\begin{theorem}
\label{theorem1}
Given a floorplan having $n$ blocks and $k$ nets having at most $t$-terminals ($t\geq2$), the algorithm STAIRoute takes $O(n^2kt)$ time.
\end{theorem}

\begin{proof}
From Lemma \ref{lemma3b}, we say that GSRG construction takes $O(t)$ time. For each 2-terminal net routing, finding the Source vertex takes $O(n)$ and our 
implementation of Dijkstra's single source shortest path algorithm (DijkstraSSP) takes $O(n^2)$. 

Again, for  $t$-terminal ($t>2$) nets, computing $G_{c}^i$ takes $O(t^2)$ and our implementation of Prim's algorithm takes $O(t^2)$. For each terminal pair ($t_i,t_j$), we obtain the shortest path using DijkstraSSP in $O(n^2)$ time. Thus, for each $t(\geq 2)$ terminal net, the time complexity is $O(t + t^2 + n^2t)$, i.e., $O(n^2t)$, since a given net may be connected to all $n$ blocks resulting in $t = n$ in the worst case. Therefore, the overall worst case time complexity for all $k$ nets is $O(n^2kt)$.
\end{proof}

\section{Experimental Results}
\label{sec:result}
We have implemented the proposed algorithm \textit{STAIRoute} in C programming language and run on a 64bit Linux platform powered by Intel Core2 Duo (1.86GHz) and 2GB RAM. We used source code for top-down hierarchical monotone staircase bipartitioning algorithm implemented by \cite{karb} to obtain the BAG and MSC tree data structure. We used MCNC/GSRC hard floorplanning benchmark circuits as given in Table \ref{tab:bench}. In order to test our algorithm, four different instances for each of the benchmarks were generated with a random seed using \textit{Parquet} \cite{adya,parque} tool. For a given circuit, the \textit{best case} (BC) and the \textit{worst case} (WC) instances among a set of floorplan topologies are solely designated in the context of total \textit{half perimeter wire length} (HPWL) of all the nets, as the ones with the smallest and the largest HPWL respectively. 
\begin{table}[!ht]
\centering{
\processtable{Floorplanning Benchmarks: MCNC and GSRC Circuits \label{tab:bench}}
{\begin{tabular*}{0.6\textwidth}{@{\extracolsep{\fill}}|c|r|r|r|r|}\toprule
{Suite} & {Circuit} & {\#Blocks} & {\#Nets}  &  {Avg. Net-deg}\\ 
\midrule
{MCNC} & apte & 9 & 44 & 3.500\\
& hp & 11 & 44 & 3.545\\  
& xerox & 10 & 183 & 2.508\\ 
& ami33 & 33 & 84 & 4.154\\  
& ami49 & 49 & 377 & 2.337\\  \hline
{GSRC} & n10 & 10 & 54 & 2.129\\  
& n30 & 30 & 147 & 2.102\\ 
& n50 & 50 & 320 & 2.112\\  
& n100 & 100 & 576 & 2.135\\ 
& n200 & 200 & 1274 & 2.138\\ 
& n300 & 300 & 1632 & 2.161\\ \hline
\end{tabular*}}{}}
\end{table}

To the best of our knowledge, there exists no such early routability estimator at the floorplanning stage. Therefore, we believe that it is unfair to have a direct comparison of our early global routing results on a given floorplan with the existing post-placement global routing results specially by those using pattern routing such as \cite{kast, zcao, ychang}, as they are in the different scopes of the flow. Therefore, we are restricted to compare our results with the state-of-the-art Steiner tree base routing topology estimator tool FLUTE \cite{cchu}.

We present our experiments by running the proposed global routing method using HV reserved layer model restricted up to eight metals layers. We consider BC and WC floorplan topologies for each of the circuits. In these experiments, we refer to Figure \ref{fig:capscale} (a) for $3$ different capacity scaling profiles, and also refer to Figure \ref{fig:routedir} for two possible directions for exploring the possible routing paths. We iterate that all the results presented here correspond to $100\%$ routability and restricted to \textit{Under-Congestion} region of Figure \ref{fig:cong} that ensures no congestion in any routing resource. Otherwise, $100\%$ routability can not be achieved.

Following are the configurations considered for conducting the experiments on the benchmark circuits given in Table \ref{tab:bench}:
\begin{enumerate}
 \item \textit{Forward} search method with \textit{No Capacity Scaling} ($FCN$)
 \item \textit{Forward} search method with \textit{Hyperbolic Capacity Scaling} ($FCH$)
 \item \textit{Forward} search method with \textit{Ladder type Capacity Scaling} ($FCL$) 
 \item \textit{Backward} search method with \textit{No Capacity Scaling} ($BCN$)
 \item \textit{Backward} search method with \textit{Hyperbolic Capacity Scaling} ($BCH$)
 \item \textit{Backward} search method with \textit{Ladder type Capacity Scaling} ($BCL$) 
\end{enumerate}

In Figure \ref{fig:n300-plot}, we showcase the routing results for $n300$ for all six run configurations and two different floorplan instances, namely BC and WC. While studying these plots, we notice that the forward (backward) search with hyperbolic scaling $FCH$ ($BCH$) gives the worst results as compared to the other two configurations \{$FCN,FCL$\} (\{$BCN,BCL$\}) both in terms of routed wirelength and via count, both in case of BC and WC floorplans for a given circuit. This is due to the fact that the hyperbolic profile is the most stringent profile among the other profiles as depicted in Figure \ref{fig:capscale} (a). We also notice similar pattern for runtime as well as presented in Figure \ref{fig:n300-plot} (c).
\begin{figure}[!ht]
\begin{subfigure}[b]{0.48\textwidth}
\centering{
\includegraphics[scale=0.21]{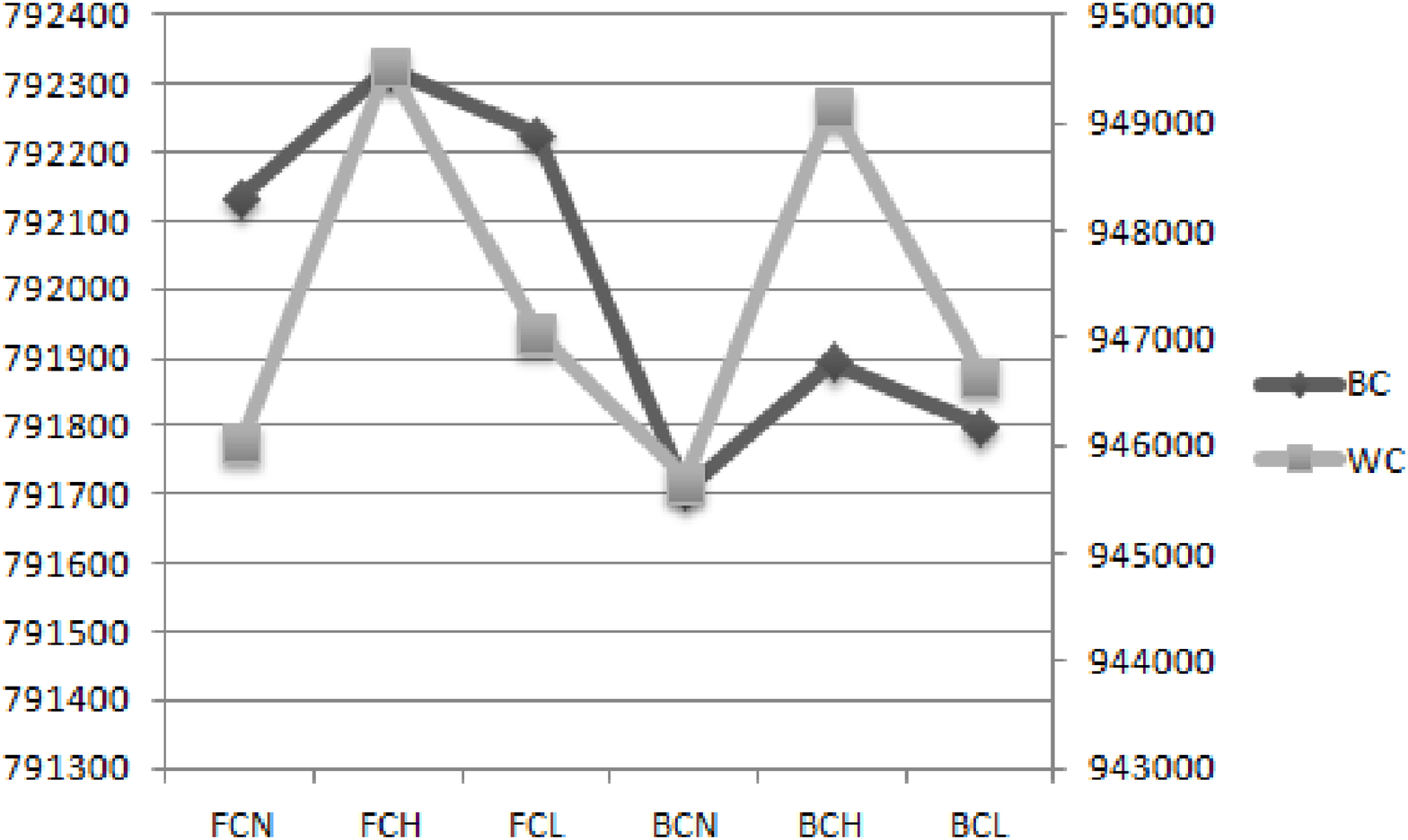}}
\caption{}
\end{subfigure}
\begin{subfigure}[b]{0.48\textwidth}
\centering{
\includegraphics[scale=0.21]{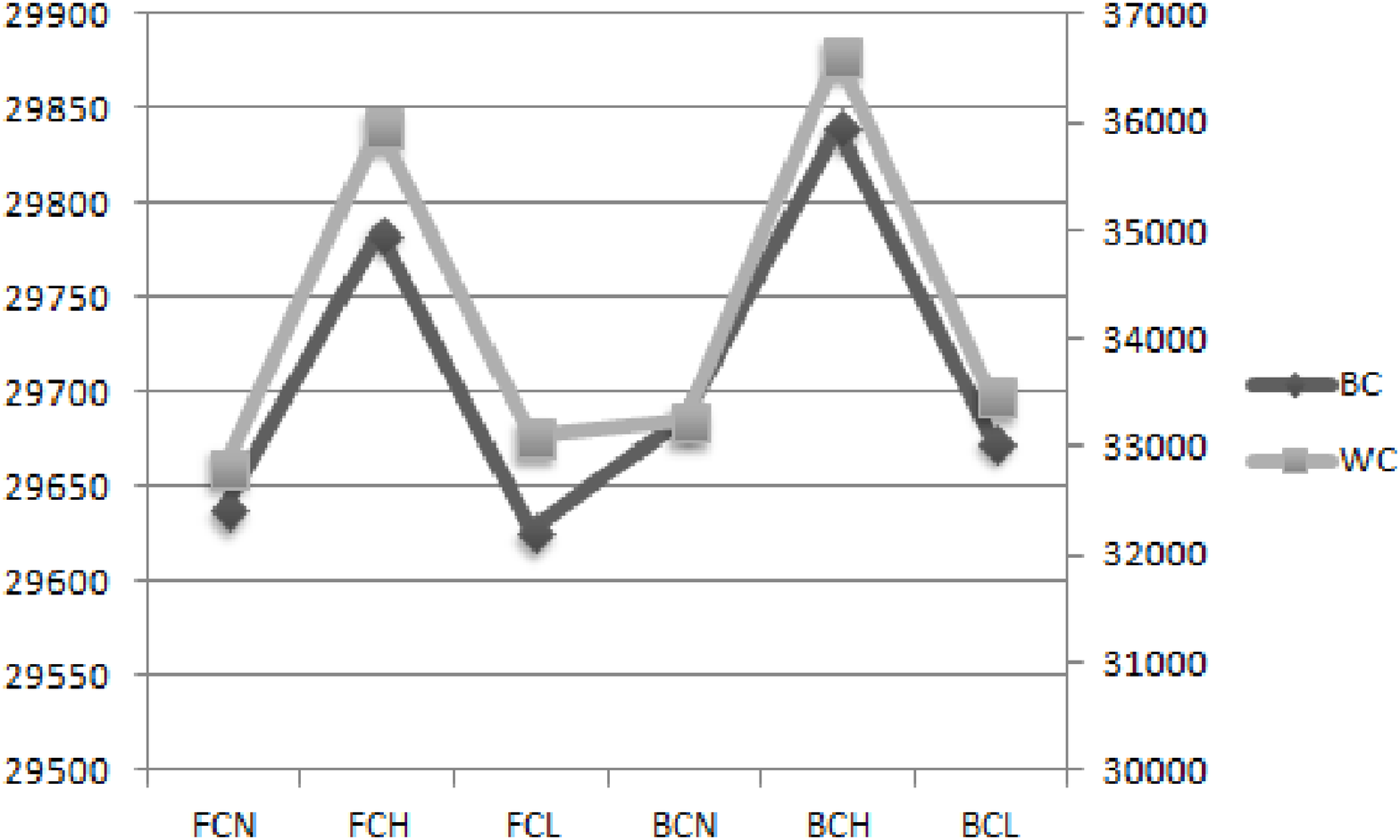}}
\caption{}
\end{subfigure}
\begin{subfigure}[b]{0.48\textwidth}
\centering{
\includegraphics[scale=0.21]{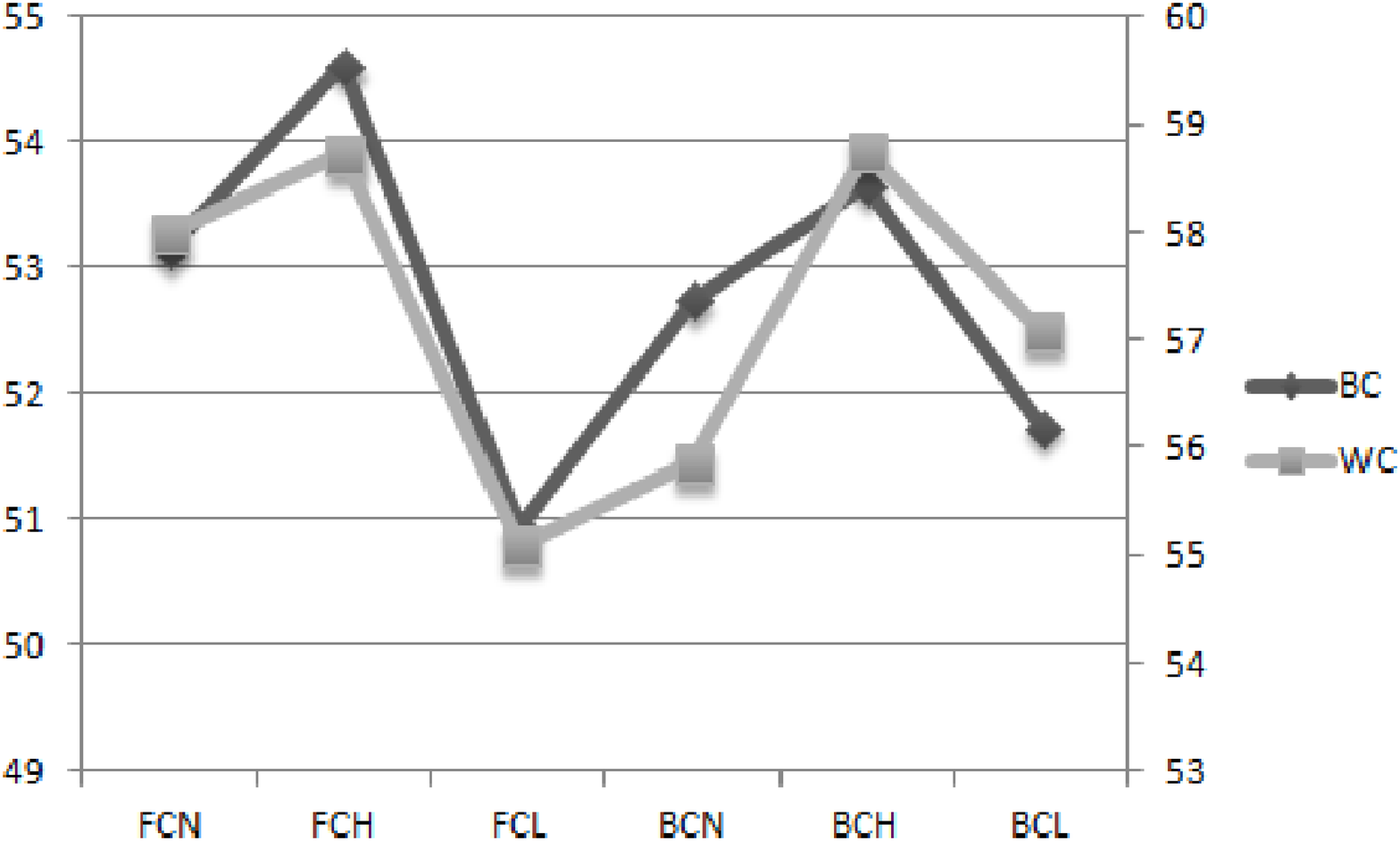}}
\caption{}
\end{subfigure}
\caption{Plot for \textit{best case} (BC) and \textit{worst case} (WC) floorplan instances of $n300$ vs. different run configurations: (a) Wirelength  ($\mu$m), (b) Via Count, and (c) Runtime(sec) \label{fig:n300-plot}}
\end{figure}

Next, we focus on the corresponding results obtained for the remaining configurations \{$FCN,FCL$\} (\{$BCN,BCL$\}) for both BC and WC floorplan topologies. Figure \ref{fig:n300-plot} (a) shows that $FCN$ ($BCN$) gives better wirelength against $FCL$ ($BCL$) both in BC and WC. Although it reflects a similar trend in the respective via count for the WC topology, $FCL$ ($BCL$) has better via count as compared to $FCN$ ($BCN$) in BC (Figure \ref{fig:n300-plot} (b)). We conduct another set of comparison for wirelength and via count between $FCN$ and $BCN$ ($FCL$ and $BCL$) for both BC and WC. Although backward search produces better wirelength as compared to that in forward search method, it incurs more vias to route a set of nets than its counterpart. This clearly shows that an early routing solution not only depends on the search direction and the capacity profiles, but also on different floorplan instances of the same circuit.  
\begin{table*}[!ht]
\fwprocesstable{Wirelength for different run configurations against FLUTE \cite{cchu} length: for BC floorplan instances only \label{tab:length-comp-bc}}
{\begin{tabular*}{\textwidth}{@{\extracolsep{\fill}}|c|r|r|r|r|r|r|r|}\toprule
      \textbf{Circuit} & \textbf{$FCN$($\mu$m)} & \textbf{$FCH$($\mu$m)} & \textbf{$FCL$($\mu$m)} & \textbf{$BCN$($\mu$m)} & \textbf{$BCH$($\mu$m)} & \textbf{$BCL$($\mu$m)} & \textbf{FLUTE($F$)} \\
       & \textbf{($FCN/F$)} & \textbf{($FCH/F$)} & \textbf{($FCL/F$)} & \textbf{($BCN/F$)} & \textbf{($BCH/F$)} & \textbf{($BCL/F$)} & ($\mu$m)  \\
      \midrule
      apte & 398137.031 & 398411.250 & 398137.031 & 396034.313 & 396308.563 & 396034.313 & 338628.000 \\

       & (1.1757) & (1.1765) & (1.1757) & (\textbf{1.1695}) & (1.1703) & (\textbf{1.1695}) &  \\
\hline
      hp & 201996.672 & 203060.359 & 201996.672 & 211536.609 & 210086.344 & 211536.609 & 123716.000 \\

       & (\textbf{1.6327}) & (1.6413) & (\textbf{1.6327}) & (1.7099) & (1.6981) & (1.7099) &  \\
\hline
      xerox & 716916.688 & 717291.625 & 716916.688 & 710575.375 & 710575.375 & 710575.375 & 633533.000 \\

       & (1.1316) & (1.1322) & (1.1316) & (\textbf{1.1216}) & (\textbf{1.1216}) & (\textbf{1.1216}) &  \\
\hline
      ami33 & 111126.102 & 111563.070 & 111126.102 & 111481.008 & 111896.906 & 111481.008 & 92330.000 \\

       & (\textbf{1.2036}) & (1.2083) & (\textbf{1.2036}) & (1.2074) & (1.2119) & (1.2074) &  \\
\hline
      ami49 & 1925760.875 & 2009223.625 & 1925760.875 & 1926176.750 & 2010630.125 & 1926176.750 & 1608746.000 \\

       & (\textbf{1.1971}) & (1.2489) & (\textbf{1.1971}) & (1.1973) & (1.2498) & (1.1973) &  \\
\hline
      n10 & 19837.000 & 19837.000 & 19837.000 & 18497.500 & 18497.500 & 18497.500 & 16626.000 \\

       & (1.1931) & (1.1931) & (1.1931) & (\textbf{1.1126}) & (\textbf{1.1126}) & (\textbf{1.1126}) &  \\
\hline
      n30 & 59585.000 & 59585.000 & 59585.000 & 58761.500 & 58761.500 & 58761.500 & 49370.000 \\

       & (1.2069) & (1.2069) & (1.2069) & (\textbf{1.1902}) & (\textbf{1.1902}) & (\textbf{1.1902}) &  \\
\hline
      n50 & 151604.000 & 152119.000 & 151851.000 & 150741.000 & 151256.000 & 150988.000 & 125018.000 \\

       & (1.2127) & (1.2168) & (1.2146) & (\textbf{1.2058}) & (1.2099) & (1.2077) &  \\
\hline
      n100 & 251455.500 & 252647.500 & 251476.500 & 250771.000 & 251917.000 & 250792.000 & 212112.000 \\

       & (1.1855) & (1.1911) & (1.1856) & (\textbf{1.1823}) & (1.1877) & (1.1824) &  \\
\hline
      n200 & 429854.500 & 430043.500 & 429923.500 & 428987.000 & 429128.000 & 429056.000 & 381021.000 \\

       & (1.1282) & (1.1287) & (1.1283) & (\textbf{1.1259}) & (1.1263) & (1.1261) &  \\
\hline
      n300 & 792135.000 & 792323.000 & 792229.000 & 791707.500 & 791895.500 & 791801.500 & 699006.000 \\

       & (1.1332) & (1.1335) & (1.1334) & (\textbf{1.1326}) & (1.1329) & (1.1328) &  \\
\botrule
\end{tabular*}}{}
\end{table*}

In Table \ref{tab:length-comp-bc} (\ref{tab:length-comp-wc}), we summarize the wirelength obtained for all the circuits in all the configurations for the respective BC (WC) topologies and compare them with the corresponding FLUTE \cite{cchu} length. For a given configuration, the corresponding wirelength is accompanied by its ratio of wirelength and FLUTE length, e.g., $FCN/F$ in the bracket below it and the best ratio(s) are highlighted. It is evident that from these results that, except for the circuits with $100$ or more blocks, the length ratio pairs $FCN/F$ and $FCL/F$ ($BCN/F$ and $BCL/F$) have little variation and also shows that $BCN$ yields the best wirelength of a given circuit with respect to the corresponding FLUTE length for most of the circuits. The best results among all the run configurations are highlighted in bold letters.
\begin{table*}[!ht]
\fwprocesstable{Wirelength for different run configurations against FLUTE \cite{cchu} length: for WC floorplan instances only \label{tab:length-comp-wc}}
{\begin{tabular*}{\textwidth}{@{\extracolsep{\fill}}|c|r|r|r|r|r|r|r|}\toprule
 \textbf{Circuit} & \textbf{$FCN$($\mu$m)} & \textbf{$FCH$($\mu$m)} & \textbf{$FCL$($\mu$m)} & \textbf{$BCN$($\mu$m)} & \textbf{$BCH$($\mu$m)} & \textbf{$BCL$($\mu$m)} & \textbf{FLUTE($F$)} \\
       & \textbf{($FCN/F$)} & \textbf{($FCH/F$)} & \textbf{($FCL/F$)} & \textbf{($BCN/F$)} & \textbf{($BCH/F$)} & \textbf{($BCL/F$)} & ($\mu$m)  \\
\midrule
   apte & 450376.469 & 450376.469 & 450376.469 & 438912.781 & 439198.781 & 438912.781 & 389806.000 \\

       & (\textbf{1.1554}) & (\textbf{1.1554}) & (\textbf{1.1554}) & (1.1260) & (1.1267) & (1.1260) &  \\
\hline
      hp & 232782.313 & 232782.313 & 232782.313 & 230255.938 & 230255.938 & 230255.938 & 144993.000 \\

       & (1.6055) & (1.6055) & (1.6055) & (\textbf{1.5880}) & (\textbf{1.5880}) & (\textbf{1.5880}) &  \\
\hline
      xerox & 1542729.875 & 1590995.125 & 1542729.875 & 1511243.250 & 1558185.250 & 1511243.250 & 1391401.000 \\

       & (1.1088) & (1.1434) & (1.1088) & (\textbf{1.0861}) & (1.1199) & (\textbf{1.0861}) &  \\
\hline
      ami33 & 120903.578 & 120903.578 & 120903.578 & 118746.414 & 118969.086 & 118746.414 & 105025.000 \\

       & (1.1512) & (1.1512) & (1.1512) & (\textbf{1.1306}) & (1.1328) & (\textbf{1.1306}) &  \\
\hline
      ami49 & 1914369.375 & 1914369.375 & 1914369.375 & 1898528.125 & 1898528.125 & 1898528.125 & 1684114.000 \\

       & (1.1367) & (1.1367) & (1.1367) & (\textbf{1.1273}) & (\textbf{1.1273}) & (\textbf{1.1273}) &  \\
      \hline
      n10 & 24526.500 & 25116.500 & 24526.500 & 23707.500 & 24297.500 & 23707.500 & 20012.000 \\

       & (1.2256) & (1.2551) & (1.2256) & (\textbf{1.1847}) & (1.2141) & (\textbf{1.1847}) &  \\
\hline
      n30 & 74743.500 & 75105.500 & 74838.500 & 74151.500 & 74513.500 & 74246.500 & 59879.000 \\

       & (1.2482) & (1.2543) & (1.2498) & (\textbf{1.2384}) & (1.2444) & (1.2399) &  \\
\hline
      n50 & 187971.000 & 189752.000 & 188476.000 & 187197.000 & 188994.000 & 187702.000 & 158173.000 \\

       & (1.1884) & (1.1996) & (1.1916) & (\textbf{1.1835}) & (1.1949) & (1.1867) &  \\
\hline
      n100 & 277304.500 & 277950.500 & 277583.500 & 276545.000 & 277251.000 & 276824.000 & 238841.000 \\

       & (1.1610) & (1.1637) & (1.1622) & (\textbf{1.1579}) & (1.1608) & (1.1590) &  \\
\hline
      n200 & 836136.000 & 865590.000 & 837493.000 & 835535.500 & 864329.500 & 836928.500 & 749479.000 \\

       & (1.1156) & (1.1549) & (1.1174) & (\textbf{1.1148}) & (1.1532) & (1.1167) &  \\
\hline
      n300 & 946039.500 & 949547.500 & 947076.500 & 945688.500 & 949165.500 & 946675.500 & 830035.000 \\

       & (1.1398) & (1.1440) & (1.1410) & (\textbf{1.1393}) & (1.1435) & (1.1405) &  \\

\botrule
\end{tabular*}}{}
\end{table*}

In Table \ref{tab:via-comp}, we present the via count for all the run configurations for each circuit for BC (WC in brackets) instances. The results clearly point out the consequence of forward and backward search on via count. It is also evident that via count in case of $FCH$ ($BCH$) is the worst among all other two configurations, namely \{$FCN, FCL$\} (\{$BCN, BCL$\}), as we have seen in case of wirelength. There is little variation in via count for relatively smaller circuits in case of \{$FCN,FCL$\} (\{$BCN,BCL$\}) and becomes significant for relatively larger circuits. The best via count (highlighted with bold letters) for each circuit in both BC and WC (in brackets) is mostly obtained in case of \{$FCN, FCL$\}.
\begin{table}[!ht]
\scriptsize
\processtable{Via count for different run configurations: for BC (WC in brackets)\\ floorplan instance \label{tab:via-comp}}
    {\begin{tabular*}{0.7\textwidth}{@{\extracolsep{\fill}}|c|r|r|r|r|r|r|}\toprule
      \textbf{Circuit} & \textbf{$FCN$} & \textbf{$FCH$} & \textbf{$FCL$} & \textbf{$BCN$} & \textbf{$BCH$} & \textbf{$BCL$} \\
      \midrule
      apte & \textbf{404} & 508 & \textbf{404} & 412 & 504 & 412 \\
           & (\textbf{452}) & (660) & (\textbf{452}) & (460) & (656) & (460) \\
      \hline
      hp & \textbf{502} & 720 & \textbf{502} & 536 & 730 & 536 \\
         & (\textbf{430}) & (630) & (\textbf{430}) & (430) & (626) & (430) \\
      \hline
      xerox & \textbf{1190} & 1238 & \textbf{1190} & 1220 & 1272 & 1220 \\
            & (\textbf{1401}) & (2261) & (\textbf{1401}) & (1527) & (2369) & (1547) \\
      \hline
      ami33 & \textbf{1156} & 1500 & \textbf{1156} & 1162 & 1530 & 1162 \\
            & (1240) & (1528) & (1240) & (\textbf{1234}) & (1586) & (\textbf{1234}) \\
      \hline
      ami49 & \textbf{3290} & 4466 & \textbf{3290} & 3406 & 4676 & 3406 \\
            & (\textbf{3629}) & (3789) & (\textbf{3629}) & (3877) & (4137) & (3877) \\
      \hline
      n10 & \textbf{176} & \textbf{176} & \textbf{176} & \textbf{176} & \textbf{176} & \textbf{176} \\
          & (\textbf{220}) & (278) & (\textbf{220}) & (222) & (260) & (222) \\
      \hline
      n30 & \textbf{937} & 941 & \textbf{937} & 956 & 960 & 956 \\
          & (1047) & (1100) & (\textbf{1014}) & (1059) & (1108) & (1026) \\
      \hline
      n50 & 3194 & 3609 & 3184 & 3191 & 3598 & \textbf{3181} \\
          & (\textbf{3252}) & (5089) & (3296) & (3402) & (5078) & (3442) \\
      \hline
      n100 & \textbf{6748} & 7553 & \textbf{6748} & 6795 & 7623 & 6799 \\
           & (\textbf{7742}) & (10050) & (7746) & (7846) & (10146) & (7850) \\
      \hline
      n200 & 18016 & 18040 & 18008 & 17977 & 17993 & \textbf{17969} \\
           & (\textbf{16905}) & (26828) & (17610) & (17253) & (27041) & (17571) \\
      \hline
      n300 & 29639 & 29785 & \textbf{29627} & 29687 & 29841 & 29675 \\
           & (\textbf{32814}) & (35955) & (33093) & (33235) & (36624) & (33461) \\ \hline
\end{tabular*}}{}
\end{table}

In our congestion analysis, we focus on relative congestion (ration of routing demand $u$ and capacity $r$) instead of absolute congestion measured by total overflow in any routing resource (edge in case of grid graph model). This leads us to follow the method prescribed in \cite{wei} in order to analyze the global congestion scenario in a given floorplan instance. The authors in \cite{wei} used relative congestion (ratio of routing demand and routing capacity) in any edge of grid graph model instead of absolute congestion, i.e., estimating total overflow by computing excess routing demand with respect to the capacity over all edges. They proposed a new metric called \textit{average congestion per edge} based on relative congestion for a certain percentage of top congested edges among all congested edges. This is denoted as $ACE(x\%)$ where $x$ is the percentage value of the worst congested edges. Based on this metric, an weighted average of $ACE(x\%)$ for four different values of $x = {0.5,1,2,5}$ is computed and denoted as $wACE4$. In this work, we adopt a similar way to compute related congestion as \textit{normalized usage} (vide Equation \ref{eqn1a}).

In Table \ref{tab:cong-comp}, we capture the $wACE4$ values among all layers for each circuit for all the run configurations in BC and WC (in brackets) to showcase the corresponding congestion scenario when $100\%$ routability is achieved. The results validate that our method conforms to the congestion profile depicted in Figure \ref{fig:cong} and does not cross the value of $1.0$ in any case. The results presented in this table correspond to the maximum $wACE4$ value for a given circuit and the one with best value among all the configurations is highlighted with bold.
\begin{table}[!ht]
\scriptsize
\processtable{Congestion (wACE4) for different run configurations: for BC (WC in brackets)\\ floorplan instance \label{tab:cong-comp}}
     {\begin{tabular*}{0.7\textwidth}{@{\extracolsep{\fill}}|c|r|r|r|r|r|r|}\toprule
      \textbf{Circuit} & \textbf{$FCN$} & \textbf{$FCH$} & \textbf{$FCL$} & \textbf{$BCN$} & \textbf{$BCH$} & \textbf{$BCL$} \\
      \midrule
      apte & 0.9011 & 0.8384 & 0.9011 & 0.9449 & 0.9063 & 0.9449 \\
       & (0.8738) & (0.9861) & (0.8738) & (\textbf{0.7971}) & (0.8750) & (\textbf{0.7971}) \\
      \hline       
      hp & 0.9914 & 0.9871 & 0.9914 & 0.9871 & 0.9853 & 0.9871 \\
       & (0.9906) & (0.8135) & (0.9906) & (0.9906) & (\textbf{0.8021}) & (0.9906) \\
      \hline    
      xerox & 0.6229 & 0.6139 & 0.6229 & 0.6229 & \textbf{0.6090} & 0.6229 \\
       & (0.9940) & (0.9375) & (0.9940) & (0.9583) & (0.9375) & (0.9583) \\
      \hline
      ami33 & 0.9673 & \textbf{0.6874} & 0.9673 & 0.9712 & 0.6920 & 0.9712 \\
       & (0.9388) & (0.9944) & (0.9388) & (0.9120) & (0.8102) & (0.9120) \\
      \hline
      ami49 & 0.9750 & 0.9931 & 0.9750 & 0.9750 & 0.9976 & 0.9750 \\
       & (0.7129) & (0.7215) & (0.7129) & (0.7115) & (\textbf{0.7068}) & (0.7115) \\
      \hline
      n10 & \textbf{0.6435} & \textbf{0.6435} & \textbf{0.6435} & \textbf{0.6435} & 0.7670 & \textbf{0.6435} \\
       & (0.9625) & (0.9500) & (0.9625) & (0.9625) & (0.9500) & (0.9625) \\
      \hline       
      n30 & 0.6585 & 0.9826 & 0.6585 & \textbf{0.6570} & 0.9788 & \textbf{0.6570} \\
       & (0.9547) & (0.9667) & (0.9547) & (0.9546) & (0.9500) & (0.9546) \\
      \hline
      n50 & 0.8179 & 0.6627 & 0.8209 & 0.8163 & \textbf{0.6343} & 0.8182 \\
       & (0.8399) & (0.8246) & (0.8488) & (0.8644) & (0.8351) & (0.8733) \\
       \hline
      n100 & 0.9904 & 0.9772 & 0.9904 & 0.9871 & 0.9883 & 0.9871 \\
       & (\textbf{0.9187}) & (0.9974) & (\textbf{0.9187}) & (0.9282) & (0.9984) & (0.9282) \\
      \hline
      n200 & 0.4892 & 0.8269 & 0.4892 & \textbf{0.4831} & 0.8152 & \textbf{0.4831} \\
       & (0.8750) & (0.6793) & (0.8555) & (0.8923) & (0.6463) & (0.8328) \\
      \hline
      n300 & \textbf{0.6561} & 0.9863 & \textbf{0.6561} & 0.6681 & 0.9843 & 0.6681 \\
       & (0.9813) & (0.9964) & (0.9811) & (0.9774) & (0.9903) & (0.9770) \\ \hline
\end{tabular*}}{}
\end{table}

In Table \ref{tab:runtime-comp}, we report runtime in seconds for all the circuits versus the said run configurations. These results correspond to both BC and WC (in brackets) floorplan instances. As we can see that the best (as highlighted) runtime in the context of BC and WC (in brackets) floorplan instances for a given circuit is given by either $FCL$ or $BCL$ for most of the cases.
\begin{table}[!ht]
\scriptsize
\processtable{Runtime (sec) for different run configurations: for BC (WC in brackets)\\ floorplan instance \label{tab:runtime-comp}}
 {\begin{tabular*}{0.7\textwidth}{@{\extracolsep{\fill}}|c|r|r|r|r|r|r|}\toprule
      \textbf{Circuit} & \textbf{$FCN$} & \textbf{$FCH$} & \textbf{$FCL$} & \textbf{$BCN$} & \textbf{$BCH$} & \textbf{$BCL$} \\
      \midrule
      apte & 0.114 & 0.109 & \textbf{0.103} & 0.108 & 0.109 & 0.113 \\
       & (0.112) & (0.106) & (\textbf{0.102}) & (0.103) & (0.104) & (0.103) \\
 \hline
      hp & 0.115 & 0.107 & 0.105 & 0.107 & \textbf{0.102} & 0.103 \\
       & (0.114) & (0.107) & (\textbf{0.100}) & (0.107) & (0.102) & (0.101) \\
 \hline
      xerox & 0.126 & 0.124 & 0.122 & 0.122 & \textbf{0.121} & \textbf{0.121} \\
       & (0.140) & (0.125) & (\textbf{0.116}) & (0.130) & (0.125) & (0.120) \\
 \hline
      ami33 & 0.172 & 0.167 & \textbf{0.151} & 0.161 & 0.159 & 0.156 \\
       & (0.186) & (0.178) & (\textbf{0.158}) & (\textbf{0.158}) & (0.166) & (0.170) \\
 \hline
      ami49 & 0.464 & 0.453 & \textbf{0.438} & 0.441 & 0.458 & 0.442 \\
       & (0.483) & (0.465) & (\textbf{0.446}) & (0.463) & (0.465) & (0.460) \\
 \hline
      n10 & 0.122 & \textbf{0.099} & 0.101 & 0.101 & 0.101 & 0.102 \\
       & (0.109) & (\textbf{0.101}) & (0.105) & (0.103) & (0.105) & (0.102) \\
 \hline
      n30 & 0.181 & 0.164 & 0.166 & 0.164 & \textbf{0.160} & 0.164 \\
       & (0.162) & (0.161) & (\textbf{0.153}) & (0.163) & (0.164) & (0.157) \\
 \hline
      n50 & 0.423 & 0.410 & 0.398 & 0.414 & 0.399 & \textbf{0.396} \\
       & (0.449) & (0.437) & (0.440) & (0.444) & (0.441) & (\textbf{0.432}) \\
 \hline
      n100 & 2.480 & 2.408 & 2.412 & 2.399 & 2.445 & \textbf{2.386} \\
       & (2.076) & (2.096) & (2.064) & (2.073) & (2.070) & (\textbf{2.054}) \\
 \hline
      n200 & 18.342 & 18.741 & 17.718 & 18.201 & 18.175 & \textbf{17.533} \\
       & (21.472) & (20.885) & (20.909) & (21.452) & (21.315) & (\textbf{20.625}) \\
 \hline
      n300 & 53.151 & 54.596 & \textbf{50.929} & 52.739 & 53.656 & 51.730 \\
       & (57.993) & (58.719) & (\textbf{55.090}) & (55.856) & (58.742) & (57.069) \\ \hline
\end{tabular*}}{}
\end{table}

\subsection{Results for IBM HB benchmarks}
In this paper, we also used IBM HB floorplanning benchmarks \cite{hben} in Table \ref{bench2} for verifying the proposed early global routing method STAIRoute. These benchmarks were derived from ISPD98 placement benchmark circuits with certain modifications (see \cite{hben} for details). These floorplan instances for each circuit were generated using \textit{Parquet} \cite{adya,parque} using random seed.  
\begin{table}[!ht]
\centering{
   \processtable{IBM HB Floorplanning Benchmark Circuits \cite{hben} \label{bench2}}
   {\begin{tabular*}{0.6\textwidth}{@{\extracolsep{\fill}}|c|r|r|r|r|}\toprule
      Circuit & {\#Blocks} &  {\#Nets} & {Avg.} & {HPWL} \\ 
      Name &   &   & {NetDeg} & {($10^6 \mu$m)} \\     
      \midrule
      ibm01 & 2254 & 3990 & 3.94 & 8.98 \\ \hline
      ibm02 & 3723 & 7393 & 4.84 & 22.19 \\ \hline
      ibm03 & 3227 & 7673 & 4.18 & 23.83 \\ \hline
      ibm04 & 4050 & 9768 & 3.92 & 30.82 \\ \hline
      ibm05 & 1612 & 7035 & 5.58 & 18.12 \\ \hline
      ibm06 & 1902 & 7045 & 4.92 & 21.78 \\ \hline
      ibm07 & 2848 & 10822 & 4.44 & 42.48 \\ \hline
      ibm08 & 3251 & 11250 & 4.92 & 46.57 \\ \hline
      ibm09 & 2847 & 10723 & 4.08 & 48.35 \\ \hline
      ibm10 & 3663 & 15590 & 3.85 & 121.23 \\ \hline
    \end{tabular*}}{}}
\end{table}

As mentioned earlier, we can not directly compare our results with that of the state-of-the-art post-placment global routers. As a results, we come up with an indirect approach to compare wirelength obtained for each circuit, by normalizing it with steiner length computed by FLUTE \cite{cchu}. The comparison results are summarized in Table \ref{tab-Chap5:Routing-comp} and shown that the values related STAIRoute only marginally higher. This is due to the fact that STAIRoute does not have the scope of routing the tiny nets that connect the standard cells, while the others presented in the table perform those tiny nets as well, in addition to the longer nets that can also be abstracted at the flooplanning level.
\begin{table}
\scriptsize
\centering{
    \processtable{Normalized (w.r.t Steiner length \cite{cchu}) wirelength between the existing\\
     post-placement global routers and STAIRoute \label{tab-Chap5:Routing-comp}}
    {\begin{tabular*}{0.7\textwidth}{@{\extracolsep{\fill}}|c|r|r|r|r|r|r|r|}\toprule
      Circuit & \multicolumn{6}{c}{Post-placement Global Routers} &  {STAIRoute$^c$}\\
      Name  & \cite{ozdal}$^b$ & \cite{mcho1}$^b$ & \cite{zhang2}$^b$ &  \cite{royj}$^b$ & \cite{moffit}$^b$ & \cite{ychang}$^b$ & \\
      \midrule
      ibm01 & 1.071 & 1.042 & 1.068 & 1.053 & 1.059 & 1.039 & 1.156 \\ \hline
      ibm02 & 1.036 & 1.032 & 1.038 & 1.018 & 1.027 & 1.024 & 1.175 \\ \hline
      ibm03 & 1.007 & 1.007 & 1.007 & 1.005 & 1.010 & 1.005 & 1.175 \\ \hline
      ibm04 & 1.045 & 1.028 & 1.046 & 1.027 & 1.045 & 1.023 & 1.155 \\ \hline
      ibm05$^d$ & - & - & - & - & - & - & 1.198 \\ \hline
      ibm06 & 1.011 & 1.007 & 1.013 & 1.006 & 1.013 & 1.007 & 1.166 \\ \hline
      ibm07 & 1.018 & 1.006 & 1.015 & 1.007 & 1.016 & 1.007 & 1.192 \\ \hline
      ibm08 & 1.005 & 1.008 & 1.009 & 1.006 & 1.010 & 1.006 & 1.197 \\ \hline
      ibm09 & 1.007 & 1.006 & 1.009 & 1.004 & 1.011 & 1.008 & 1.199 \\ \hline
      ibm10 & 1.016 & 1.027 & 1.015 & 1.008 & 1.020 & 1.010 & 1.187 \\ \hline
      Average  & 1.024 & 1.018 & 1.024 &1.015 & 1.024 & 1.014 & 1.180 \\ \hline
    \end{tabular*}}{}}
  {\scriptsize $b$ - using ISPD98 global routing benchmarks, $c$ - using IBM-HB floorplanning benchmarks, and $d$ - no result on $ibm05$ of ISPD98 benchmark by the existing global routers}
\end{table}

\section{Conclusion}
\label{sec:discuss}
In this paper, we proposed a novel early global routing framework based on monotone staircase regions obtained by hierarchical bipartitioning of a given flooplan instances of a circuit. It thus immediately follows the floorplanning stage in the existing VLSI implementation flow and hence require no detailed placement of standard cells in order to conduct global routing of the nets available at this level of design abstraction. Unlike the existing global routers following the standard cell placement stage, the monotone staircase regions act as the routing resources and the nets are routed though them for a given number of metal layers. Both unreserved or reserved layer model can be used in this framework. The congestion scenario is modeled in this routing model in such way that the utilization is no more than $100\%$ in any of the routing resources, by switching to next permissible layer. Multi-terminal net decomposition using the proposed Steiner tree method is unique and no other existing methods are known to have a similar model. Our experimental results show that $100\%$ routing completion is possible without any congestion even for different capacity profiles that include constrained metal pitch/width variation due to recent fabrication processes. Additionally, employing different search directions shows that improvement in routed wirelength and via count may potentially be achieved. 

The proposed framework has a two fold potential advantage: (a) evaluate the given floorplan on the basis of early global routing results obtained by STAIRoute, and (b) estimate the global routing metrics and related information essential for the subsequent stages of the design flow. This work may be extended to incorporate design for manufacturability (DFM) issues by suitably modeling them into the routing cost.



\end{document}